\documentclass[twocolumn]{autart}    

\makeatletter
\let\c@author\relax
\makeatother

\usepackage{graphics}
\usepackage{epsfig}
\usepackage{pgfplots}
\usepackage{amsmath, amssymb, amsfonts,eucal,amsbsy,bm}
\usepackage{todonotes}

\usepackage[giveninits=true,doi=false,isbn=false,url=false,maxbibnames=99]{biblatex}

\newcommand{\R}{\mathbb{R}}

\newcommand{\Zc}{\CMcal{Z}}

\newcommand{\U}{\CMcal{U}}
\newcommand{\W}{\CMcal{W}}

\newtheorem{proof}{Proof}[section]
\newtheorem{theorem}{Theorem}[section]
\newtheorem{lemma}{Lemma}[section]
\newtheorem{corollary}{Corollary}[section]

\theoremstyle{definition}
\newtheorem{definition}{Definition}[section]

\begin{document}

\begin{frontmatter}

\title{Blind Identification of Fully Observed Linear Time-Varying Systems via Sparse Recovery}


\author[EECS]{Roel Dobbe}\ead{dobbe@berkeley.edu},    
\author[EECS]{Stephan Liu}\ead{stephan.x.liu@berkeley.edu},               
\author[HUST]{Ye Yuan}\ead{ye.yuan@outlook.com},  
\author[EECS]{Claire Tomlin}\ead{tomlin@eecs.berkeley.edu}  

\address[EECS]{Department of Electrical Engineering \& Computer Sciences, UC Berkeley, USA}  
\address[HUST]{School of Automation, Huazhong University of Science and Technology, People's Republic of China}


\begin{abstract}                          
Discrete-time linear time-varying (LTV) systems form a powerful class of models to approximate complex dynamical systems with nonlinear dynamics for the purpose of analysis, design and control. 
Motivated by inference of spatio-temporal dynamics in breast cancer research, we propose a method to efficiently solve an identification problem for a specific class of discrete-time LTV systems, in which the states are fully observed and there is no access to system inputs. In addition, it is assumed that we do not know on which states the inputs act, which can change between time steps, and that the total number of inputs is sparse over all states and over time. 
The problem is formulated as a compressive sensing problem, which incorporates the effect of measurement noise and which has a solution with a partially sparse support.
We derive sufficient conditions for the unique recovery of the system model and input values, which lead to practical conditions on the number of experiments and rank conditions on system outputs.
Synthetic experiments analyze the method's sensitivity to noise for randomly generated models.
\end{abstract}

\end{frontmatter}



\section{Introduction}
\label{sec:introduction}

Many complex dynamical systems, such as power grids or biological systems, exhibit nonlinear dynamics.
Unfortunately, the formulation of nonlinear system identification is generally hard or intractable unless system structure can be exploited or an efficient black box model structure is used to approximate the system's dynamics \cite{nelles_nonlinear_2013}.
Nonlinear system dynamics are therefore often approximated by piecewise-affine or \emph{discrete-time linear time-varying} (LTV) models.
In addition, efforts to identify a dynamical model for analysis or control design can be hampered for certain systems or applications, due to a lack of access to the \emph{inputs} (or \emph{disturbances}) entering the system.
An example is detection and mitigation of malicious attacks on cyber-physical systems \cite{pasqualetti_attack_2013}, or inferring temporal protein-protein interactions in gene regulatory networks (GRNs).
The \emph{blind system identification problem} (BSI) assumes that the values of inputs and disturbances are unknown and aims to retrieve these together with the parameters of a dynamical system model from measured outputs of the system dynamics~\cite{hua_blind_2002}.
This problem is inherently challenging and typically requires exploiting structure of the system dynamics. Here, we apply the Occam's razor principle by exploiting sparsity in identification; not in the parameterization of the dynamic model (as generally accepted~\cite{chiuso_bayesian_2012}), but in parameterizing the unknown inputs of the BSI problem.

The approach is primarily motivated by inference problems in breast cancer biology that aim to capture protein-protein interactions in gene regulatory networks (GRNs); whether these exist, how strong these are and how these change over time \cite{heiser_et_al._subtype_2012}. 
Improving our understanding of how drugs and mutations affect GRNs is critical for effective and personalized treatment design. 
As such, the aim of this work is to retrieve both the input effects of drugs, the interaction dynamics between proteins, and how these change over time.
These effects tend to be nonlinear and time-varying. 
In earlier work, a discrete-time linear time-varying modeling structure was used to approximate these dynamics \cite{dobbe_heterogeneity_2015}. 
Here, we had access to measurements of concentrations for \emph{all} proteins in the GRN, so that the output represents a \emph{fully observed} state vector at every time instance. This assumption is increasingly practical in biological experiments due to rapid developments of measurement technologies, such as Reverse Phase Protein Array (RPPA), which allow high sensitivity and sample throughput of protein level measurements at a reasonable cost per sample \cite{akbani_realizing_2014}.
In such settings, the system input, which is the effect that a drug or a mutation has on a GRN, is typically unknown. However, it is generally true that the input is \emph{sparse}; it affects a relatively small number of states during a small number of time steps.
 
 
%
%

\subsection*{Prior Work}
BSI is known to be a difficult problem that is generally ill-posed. BSI of time-varying systems is known as a challenging problem, as compared to time-invariant systems.
It is well known that in order to reliably retrieve the input and/or system parameters, further information about the system is needed~\cite{hua_blind_2002}. 
Hence, all prior works in different areas all impose a certain structure that is rich enough to represent complex system dynamics and simple enough to allow for identification.

Originally, blind identification was well studied for Finite Impulse Response (FIR) systems, in which the filter represents the system's dynamics. For an extensive overview of BSI for time-invariant systems, the reader is referred to~\cite{abed-meraim_blind_1997,hua_blind_2002}. 
More recently, the advent of ubiquitous sensing and data collection has spurred new efforts to perform BSI for larger-scale multi-input, multi-output systems~\cite{ohlsson_blind_2014,segarra_blind_2017}.
The BSI literature for systems with time-varying dynamics is sparse. Typically these systems are modeled as an extension of time-invariant FIR filters using a basis extension approach~\cite{tsatsanis_subspace_1997,tugnait_linear_2002,giannakis_basis_1998}. 
These results are solely for single-input, multi-output systems. In addition, FIR model cannot model feedback dynamics, which requires infinite impulse response (IIR) models. Unfortunately, IIR systems driven by unknown inputs are inherently not identifiable~\cite{hua_blind_2002}.
In control theory, identification methods have been proposed for certain classes of time-varying systems that are restricted to certain structures and parameter changes.
In \cite{liu_identification_1997}, a discrete-time LTV state space model is identified, assuming no input and stable dynamics.
A widely studied approach is that of Linear Parameter-Varying (LPV) systems, for which identification procedures are proposed by \cite{lee_identification_1999}, \cite{verdult_subspace_2002} and \cite{bamieh_identification_2002}, and for which \cite{toth_modeling_2010} provides a broad and rigorous overview.
These approaches tend to be hard to scale \cite{toth_modeling_2010}, and none of these time-varying methods consider scenarios with unknown inputs.

In network inference, sparse recovery theory has been applied and further developed to reconstruct networks from data exploiting the sparsity in network connectivity \cite{napoletani_reconstructing_2008,yuan_robust_2011,zavlanos_inferring_2011,hayden_sparse_2016,chang_data-driven_2012}. 
Most work in this area assumes a linear time-invariant (LTI) model that governs the dynamic propagation of signals, often without any external inputs. \cite{sanandaji_compressive_2011} develops a method for inferring autoregressive models with exogenous inputs (ARX), in which the parameter vector changes a limited number of times. 
The use and identification of time-varying graphical models are proposed in \cite{ahmed_recovering_2009,kolar_estimating_2010,lebre_statistical_2010,wang_time_2011}.
\cite{casteigts_time-varying_2012} gives a general overview of how time-varying graphs and dynamic networks are used in different fields and application. 
In \cite{ahmed_recovering_2009},  a $\ell_1$-regularized logistic regression formalism is used to capture network structure and its changes over time. 
While scaling well to larger networks, the method does not consider the effect of external inputs.
\cite{lebre_statistical_2010} introduces auto-regressive time-varying models to describe and infer gene-regulation networks and infers the model using a Reversible Jump Markov Chain Monte Carlo procedure. 
This model class neatly encodes the time-varying dynamics with an LTV mapping, but does not consider the effect of external inputs.

\subsection*{Contributions}
We propose a blind system identification method for discrete-time LTV dynamical systems with four important characteristics: multiple inputs that have a sparse effect on the system state and over time, internal feedback dynamics, a fully observed state vector, and repeated experiments.
The first two characteristics address relevant open challenges in the literature of blind identification of time-varying systems. 
The third characteristic is an assumption that is practical in the context of inferring GRNs, and, in some sense, represents the price to pay to overcome the complexity of the former two characteristics. 
In addition, we assume that experiments can be repeated multiple times, with the same time-varying dynamics but different input values, in order to collect sufficient data for identification.
This is a fair assumption for biological studies, in which experimental conditions can be replicated efficiently, but it can form a challenge for other applications. 

\subsection*{Notation}
Denote by $A^{\top}$ the transpose of a matrix, and by $\text{vec}(\cdot) : \R^{m \times D} \to \R^{mD}$ the function that vectorizes a matrix column-by-column. The function $\| x \|_0$  (the ``$\ell_0$-norm'') returns the number of nonzero entries in the vector $x$, which is said to be $s$-sparse if at most $s$ of its entries are nonzero: $\| x \|_0 \le s$.
We will use subscript $i = 1,\hdots,n$ to denote the $i$-th entry of a vector in $\R^n$.


\section{Problem Formulation}
\label{sec:formulation}

In this section, we first represent an experimental data set as the evolution of a dynamical system using an LTV modeling framework. 
We then formulate the system identification and input retrieval problem in a sparse recovery framework.

The problem is formulated as an experimental design problem, with the aim to understand necessary and sufficient conditions on perturbations and collected output measurements of system dynamics that guarantee successful inference of parameters related to system dynamics and inputs/disturbances.
Consider a series of $q$ experiments, representing different perturbations and samples of the system state $z \in \R^n$. State output measurements are taken at $k_f$ moments, not necessarily equally sampled through time, but at fixed instances for all experiments.
The dynamics of the LTV system during an experiment $j$ are modeled as:
\begin{equation}
\begin{split}
z^{(j)}[k+1] &= A[k] z^{(j)} [k] + u^{(j)} [k] + w^{(j)} [k], \\
&  j \in \{1,\hdots,q\}, \  k \in \{0,\hdots,k_f-1\},
\end{split}
\label{eq:LTVdyn}
\end{equation}
where $z^{(j)} [k], u^{(j)} [k] , w^{(j)} [k] \in \mathbb{R}^n$ are the state vector, input vector and noise vector of a single experiment $j$, and $A[k] \in \mathbb{R}^{n \times n}$ is a matrix describing the dynamical interactions between the state variables for the transition from time $k$ to $k+1$. The matrices $A[k]$ are constant across experiments. 
The state vector $z^{(j)}[k]$ is assumed to be measured at all $k \in \{0,\hdots,k_f\}$.
We consider scenarios for which the noise is bounded by $\| w^{(j)}[k] \|_2 \le \eta_j, \forall j,k$.
Note that we allow for the same number of inputs as states - these are both $n$-dimensional.
The inputs can vary over the different experiments. We assume that the inputs are sparse over all $q$ experiments, that is out of all $n k_f q$ input values only $s < n k_f q$ are nonzero. 
By penalizing sparsity in our problem formulation, we will ensure that the number of nonzero inputs is low.
The central questions of this work are, given the collection of dynamic output data $z^{(j)}[k]$, how and under what conditions can we correctly infer the parameters in \eqref{eq:LTVdyn}: $A[k]$ ($n^2 k_f$ values) and the unknown inputs $u^{(j)}[k]$ ($s$ values) for $ j \in \{1,\hdots,q\}, \  k \in \{0,\hdots,k_f-1\}$?


For each time step $k$, we stack all our experiments together into matrices for states, inputs and noise vectors:
\begin{equation}
\begin{split}
\CMcal Z_k &= \left[ z^{(1)}[k] \cdots z^{(q)}[k]  \right] \in \R^{n \times q} \,, \\
\CMcal U_k &= \left[ u^{(1)}[k] \cdots u^{(q)}[k]  \right] \in \R^{n \times q} \,, \\
\CMcal W_k &= \left[ w^{(1)}[k] \cdots w^{(q)}[k]  \right] \in \R^{n \times q} \,. \\
\end{split}
\label{eq:stacking}
\end{equation}
We further organize the data and variables as:
\begin{equation}
\mathbf Z \triangleq 
\left[
\begin{array}{c}
\Zc_1 \\
\Zc_2 \\
\vdots \\
\Zc_{k_f} \\
\end{array}
\right] \ , \ 
\mathbf U \triangleq 
\left[
\begin{array}{c}
\U_0 \\
\U_1 \\
\vdots \\
\U_{k_f-1} \\
\end{array}
\right] \ ,
\mathbf W \triangleq 
\left[
\begin{array}{c}
\W_0 \\
\W_1 \\
\vdots \\
\W_{k_f-1} \,. \\
\end{array}
\right]
\end{equation}
We then vectorize $\mathbf Z$, $\mathbf U$, $\mathbf W$ and $A[k]$ as follows: $\pmb{z} = \text{vec}(\mathbf Z) \in \R^{n k_f q}$, $\pmb{u} = \text{vec}(\mathbf U) \in \R^{n k_f q}$, $\pmb{w} = \text{vec}(\mathbf W) \in \R^{n k_f q}$ and $\pmb{a} = \text{vec}(A[0], \cdots, A[k_f-1]) \in \R^{n^2 k_f}$ (or $\R^{n^2}$ for LTI models). 
The $i$-th scalar output in \eqref{eq:LTVdyn} can be rewritten as 
\begin{equation}
z_i^{(j)}[k+1] = \displaystyle \sum_{l = 1}^{n} a_{il}[k] z_l^{(j)}[k] + u_i^{(j)}[k] + w_i^{(j)}[k] \,,
\end{equation}
for $i = 1,\hdots,n$, $j = 1, \hdots, q$ and $k = 0,\hdots, k_f$. Here, $a_{il}[k]$ denotes the entry in $A[k]$ on the $i$-th row and $l$-th column.
Equivalently, the dynamics of one experiment can be formulated as
\begin{equation}
\begin{array}{l}
z^{(j)}[k+1] = \\
\left[ \ 0_{n \times n( k_f (j-1) + k - 1)} \ | \ I_{n}  \ | \ 0_{n \times n( k_f(q - j) + (k_f - k) )} \right]  \pmb{u} + \\ 
 \left[  \ 0_{n \times n^2(k-1)} \ | \ I_n \otimes (z^{(j)}[k])^{\top} \ | \ 0_{n \times n^2( k_f - k)} \right]
\pmb{a} + w^{(j)}[k] \,.
\end{array}
\label{eq:ConstrCSexp}
\end{equation}
Or in short
\begin{equation}
z^{(j)}[k+1]  =  \left[ \ \psi_u^{(j)}[k] \ | \ \psi_a^{(j)}[k] \ \right]
\left[
\begin{array}{c}
\pmb{u} \\
\pmb{a} 
\end{array}
\right] + w^{(j)}[k] \,.
\label{eq:ConstrCSexp2}
\end{equation}
Here, $\otimes$ denotes the Kronecker product. By stacking this equation vertically for all time steps $k = 1, \hdots, k_f$ and experiments $j = 1, \hdots, q$ we can construct
\begin{equation}
\begin{array}{rcl}
\pmb{z} &=& \left[ \ \Psi_u \ | \ \Psi_a \ \right] 
\left[
\begin{array}{c}
\pmb{u} \\
\pmb{a} 
\end{array}
\right] + \pmb{w}
= \Psi 
\left[
\begin{array}{c}
\pmb{u} \\
\pmb{a} 
\end{array}
\right] + \pmb{w} \,,
\end{array}
\label{eq:ConstrCS}
\end{equation}
where $\Psi_u = I_{n k_f q}$ and
\begin{equation}
\Psi_a = 
\left[
\begin{array}{c}
\text{blkdiag} \left( I_n \otimes (z^{(1)}[k])^{\top} \right) \\
\text{blkdiag} \left( I_n \otimes (z^{(2)}[k])^{\top} \right) \\
\vdots \\
\text{blkdiag} \left( I_n \otimes (z^{(q)}[k])^{\top} \right) \\
\end{array}
\right] \in \R^{n k_f q \times n^2 k_f} \,.
\label{eq:Psi_a}
\end{equation}
Here, for each experiment $j = 1, \hdots, q$, blkdiag$(\cdot)$ constructs a block-diagonal matrix with blocks $I_n \otimes (z^{(j)}[k])^{\top}$ for $k = 0, \hdots k_f - 1$.
Note that in the case of an LTI system, the block diagonal structure collapses, resulting in
\begin{equation}
\label{eq:Psi_a_lti}
\Psi_a = 
\left[
\begin{array}{c}
I_n \otimes (z^{(1)}[0])^{\top} \\
\vdots \\
I_n \otimes (z^{(1)}[k_f-1])^{\top} \\
\vdots \\
I_n \otimes (z^{(q)}[0])^{\top} \\
\vdots \\
I_n \otimes (z^{(q)}[k_f-1])^{\top} \\
\end{array}
\right] \in \R^{n k_f q \times n^2} \,.
\end{equation}
$\Psi \in \R^{n k_f q \times (n k_f q + n^2 k_f)}$ denotes the sensing matrix and $\pmb{w} \in \R^{n k_f q}$ a vector with stacked measurement noise values. By exploiting prior knowledge about the statistics of $\pmb{w}$, we can determine a bound on $\ell_2$-norm: $\| \pmb{w} \|_2 \le \bm{\eta}$, and hence we can formulate the constraint
\begin{equation}
\| \pmb{z} - \Psi_u \pmb{u} -  \Psi_a \pmb{a} \|_2 \le \eta \,.
\label{eq:CSconstr}
\end{equation}
Exploiting the sparsity of $\pmb{u}$, a compressive sensing formulation for inferring the unknowns $(\pmb{u},\pmb{a})$ with noisy measurements now reads
\begin{equation}
\label{eq:noisyCS}
\min_{\pmb{u}, \pmb{a}} \  \|  \pmb{u} \|_1 \,, \quad
\text{subject to} \quad \| \pmb{z} - \Psi_u \pmb{u} -  \Psi_a \pmb{a} \|_2 \le \bm{\eta} \,. 
\end{equation}
Note that we have assumed that the vector $\pmb{a}$, representing all parameters in the dynamics matrices $A[k]\,, k = 0,\hdots,k_f-1$, is not sparse. This is a realistic assumption, as the discrete time matrices $A[k]$ typically are integrals over some continuous dynamics representing the propagation of dynamic interactions over the state space, leading to a dense matrix even if few state interactions exist.
We therefore attempt to find a \emph{partially sparse} solution in which sparsity is only enforced on $\pmb{u}$, and not necessarily on $\pmb{a}$. 
If an application yields sparsity in $\pmb{a}$, this can be addressed by adding the $\ell_1$-norm of $\pmb{a}$, yielding a compressive sensing problem with a fully sparse support.

\section{Analysis}
\label{sec:analysis}

Consider the measurement equation
\begin{equation}
y = \Psi x + w \,, 
\label{eq:measeq}
\end{equation}
where $y \in \R^m, \bar x \in \R^D$, and $\| w \|_2 \le \eta$ is a bounded noise signal. 
In general, $m<D$ yields an underdetermined system of equations, with an infinite number of solutions. 
It turns out that if the \emph{sensing matrix} $\Psi$ adheres to certain conditions and the signal~$\bar x$ that generated the data~$y$ is sufficiently sparse, then $\bar x$ can be retrieved exactly from far fewer measurements (i.e. $m \ll D$) than asserted by the Nyquist sampling theorem \cite{bruckstein_sparse_2009}.
The sparsest solution to the underdetermined system of equations $y = \Psi x$ can be found by solving:
\begin{equation}
\displaystyle \min_{x \in \R^D} \| x \|_0 \,, \quad \text{subject to} \quad \| y - \Psi x \| \le \eta \,.
\label{eq:CSP0}
\end{equation}
\begin{lemma} (Unique retrieval of the sparsest solution \cite[Lemma 2.1]{donoho_for_2006})
If the sparsest solution to \eqref{eq:CSP0} has $\| x \|_0 = s$ and $D \ge 2s$ and all subsets of $2s$ columns of $\Psi$ are full rank, then this solution is unique.
\label{thm:uniqueness}
\end{lemma}
Notice that this Lemma assumes that a $s$-sparse data-generating signal~$x$ exists.
In general, \eqref{eq:CSP0} is a NP-hard optimization problem that is both combinatorial and non-convex, and hence impractical to solve. 
In contrast, the \emph{Basis Pursuit} method \cite{chen_atomic_2001} solves the convex relaxation of \eqref{eq:CSP0} efficiently,
\begin{equation}
\displaystyle \min_{x \in \R^D} \| x \|_1 \,, \quad \text{subject to} \quad \| y - \Psi x \| \le \eta \,.
\label{eq:CSP1}
\end{equation}
\begin{definition}
The spark of a matrix $\Psi$ is the smallest number of columns of $\Psi$ that are linearly dependent \cite{donoho_optimally_2003}, which is upper bounded by $rank(\Psi)+1$.
\end{definition}
Given this definition, the following lemma provides sufficient conditions for equivalence between the compressed sensing problem \eqref{eq:CSP0} and its convex relaxation \eqref{eq:CSP1}.
\begin{lemma} 
\label{lm:equiv}
(Spark Equivalence Condition \cite{donoho_optimally_2003})
For the system of linear equations $\Psi x = y$ ($\Psi \in \R^{m \times D}$ full-rank with $m < D$), if a solution $x$ exists obeying
\begin{equation}
\| x  \|_0 < \frac{1}{2} \text{spark} (\Psi) \,,
\end{equation}  
that solution is both the unique solution to the convex relaxation \eqref{eq:CSP1}, and the unique solution to the original NP-hard compressive sensing problem \eqref{eq:CSP0}.
\end{lemma}

Returning to our central problem \eqref{eq:noisyCS}, the support of the solution is partially sparse due to $\pmb{a}$ being a potentially dense vector. 
We determine under what conditions, the vectors $(\bar{\pmb{u}},\bar{\pmb{a}})$ that generated measurements $\pmb{z}$ can be retrieved.
Let $s_{\pmb{u}} = \| \bar{\pmb{u}} \|_0$ denote the number of nonzero entries of $\bar{\pmb{u}} \in \R^{n k_f q}$.
Denote the fraction of nonzero entries in $\bar{\pmb{u}}$ as
\begin{equation}
\rho_{\pmb{u}} \triangleq \frac{s_{\pmb{u}}}{n k_f q} \,.
\label{eq:cardcond}
\end{equation}

\begin{theorem} 
Suppose that the signal $(\bar{\pmb{u}},\bar{\pmb{a}})$ that generated the measurements~$\pmb{z}$, as in~\eqref{eq:ConstrCS}, is also the sparsest solution to~\eqref{eq:noisyCS} with $\| \bar{\pmb{a}} \|_0 = n^2 k_f$, and $\| \bar{\pmb{u}} \|_0 = s_{\pmb{u}}$ with $\rho_{\pmb{u}} \le \frac{1}{2}$. If $\Psi_a$ is full column rank, then the solution to \eqref{eq:noisyCS} is unique and equivalent to the solution of the NP-hard $\ell_0$-problem.
\label{thm:sparsity}
\end{theorem}

\begin{proof}
First note that as $\bar{\pmb{a}}$ is a dense vector, we require $\Psi_a$ to be full column rank in order for all information in $\bar{\pmb{a}}$ to be maintained.
Now assume $\bar{\pmb{a}}$ is known, define the vector $\tilde{\pmb{z}} = \pmb{z} - \Psi_a \bar{\pmb{a}}$, and rewrite \eqref{eq:noisyCS} as
\begin{equation}
\label{eq:noisyCSproof}
\begin{split}
\min_{\pmb{u}} \ & \|  \pmb{u} \|_1 \,,  \\
\text{s.t. } &  
\| \tilde{\pmb{z}} - \Psi_u \pmb{u} \|_2 \le \bm{\eta} \,. \\
\end{split}
\end{equation}
As $\Psi_u = I_{n k_f q}$ (identity matrix) is full rank and $\rho_{\pmb{u}} \le \frac{1}{2}$, all subsets of 2$s_{\pmb{u}}$ columns of $\Psi_u$ are also full rank.
Hence, following Lemma \ref{thm:uniqueness}, any ($s_{\pmb{u}} + n k_f q$)-sparse solution to the NP-hard $\ell_0$-problem is unique.
In addition, the same full-rank condition yields that $spark(\Psi_u) = n k_f q + 1$. With $\|\bar{\pmb{u}} \|_0 = s_{\pmb{u}} \le \frac{1}{2} n k_f q < \frac{1}{2}(n k_f q + 1) = \frac{1}{2} spark(\Psi_u)$, the spark equivalence condition in Lemma \ref{lm:equiv} is also satisfied, which guarantees that the solution of our $\ell_1$-relaxation in \eqref{eq:noisyCS} is equivalent to the solution of the corresponding $\ell_0$-problem.
\qed
\end{proof}

This result suggests that if a data-generating signal $(\bar{\pmb{u}},\bar{\pmb{a}})$ exists and is the sparsest signal explaining the measurements, then it will be uniquely recovered with \eqref{eq:noisyCS}, as long as less than half of the entries in $\bar{\pmb{u}}$ are nonzero and the matrix $\Psi_a$, which is determined by the measurement data, is full column rank.
If $a$ is also sparse, this will make the problem easier; the rank conditions on $\Psi_a$ is not necessary anymore, and the sparsity of the overall vector $[\pmb{u}; \pmb{a}]$ is bounded by $\frac{1}{2}(n k_f q+n^2 k_f)$.
Taking Theorem~\ref{thm:sparsity}, further inspection of $\Psi_a$ reveals conditions on the number of experiments and the measured data. 
\begin{corollary} 
\label{thm:persist}
(Rank Conditions on Output Matrices)
$\Psi_a$ being full column rank (as required by Theorem \ref{thm:uniqueness}), implies that

\begin{enumerate}
\item LTV case: For $k = 0, \hdots, k_f - 1$, the matrix $\CMcal Z_k \in \R^{n \times q}$, as defined in \eqref{eq:stacking}, is full row rank.

\item LTI case: The dynamics measured over all times, i.e. the matrix $\left[ \CMcal Z_0 \cdots \CMcal Z_{k_f} \right]$
is full row rank. 

\end{enumerate}
\end{corollary}

\begin{proof}
\begin{enumerate}
\item LTV case: Due to the structure of $\Psi_a$, as defined in \eqref{eq:Psi_a}, each $k$-th block of $n$ columns has at most $q$ rows with nonzero entries, equivalent to the matrix 
\begin{equation}
    \CMcal Z_k^{\top} = \left[ z^{(1)}[k] \cdots z^{(q)}[k]  \right]^{\top} \in \R^{q \times n}, \ \forall k = 0, \hdots, k_f \,,
\end{equation} 
as initially constructed in \eqref{eq:stacking}.
Every consecutive block of $n$ columns in $\Psi_a$ has its nonzero rows in different rows, due to the blkdiag$(\cdot)$) operation.
As such, the full column rank condition proposed for $\Psi_a$ can be reinterpreted as a full column rank condition on each block of $n$ columns, and thus on each matrix $\CMcal Z_k^{\top}$ for $k = 0, \hdots, k_f - 1$.
\item LTI case: Here, $\pmb{a}$ is only $n^2$-dimensional and the blkdiag$(\cdot)$ structure in \eqref{eq:Psi_a} disappears, resulting in \eqref{eq:Psi_a_lti}.
Each consecutive block of $n$ columns in $\Psi_a$ has at most $k_f q$ rows with nonzero entries, equivalent to a row permutation of the matrix
\begin{equation}
    \left[ \CMcal Z_0 \cdots \CMcal Z_{k_f - 1} \right]^{\top} \,.
\label{eq:stackeddata}
\end{equation}
Every consecutive block of $n$ columns in $\Psi_a$ has its nonzero rows in different rows.
As such, the full column rank condition proposed for $\Psi_a$ can be reinterpreted as a full row rank condition on \eqref{eq:stackeddata}. \qed
\end{enumerate} 
\end{proof}

The LTV condition implies that each time step we require sufficient (or persistent) excitation in the system state over all experiments $j = 1, \hdots, q$.
This result confirms that the number of experiments should at least be equal to or greater than the number of state variables, that is $q \ge n$.
The LTI condition implies that we require persistency of excitation in the system over both experiments $j = 1, \hdots, q$ and time steps $k = 0, \hdots k_f - 1$, which is, unsurprisingly, easier to satisfy than the LTV condition.
This confirms that the number of time steps times the number of experiments should at least be greater than or equal to the number of state variables, that is $k_f q \ge n$.
Corollary \ref{thm:persist} provides interesting experimental conditions that are sufficient, which are intuitive from the perspective of system identification, a field that traditionally tries to understand how many and what quality experiments are necessary to guarantee the \emph{identifiability} of a dynamical system model. The notion of \emph{persistency of excitation} covers this general challenge, and is typically used to understand if an input signal is able to excite the different dynamic modes of a system \cite{ljung_global_1994}. When inputs cannot be designed, Corollary \ref{thm:persist} can be used to check if the output measurements reflect persistent excitation, and combine the right data to construct a well-posed problem.

Note that in some settings, the condition $q > n$ can be restrictive, for instance for identifying larger systems with thousands of states. In the context of LTI systems, different researchers have addressed this challenge and showed that the use of multiple inputs per experiment can reduce the necessary number of experiment if states are excited simultaneously \cite{tegner_reverse_2003,gevers_identification_2006}. These principles were further developed for sparse LTI network identification via CS~\cite{hayden_sparse_2016}. Extending these result to LTV systems seems relevant and remains an open problem.

\section{Experimental Validation}
\label{sec:simulation}

The method is tested via synthetic experiments to study its sensitivity to noise, and to relate numerical results to the theoretical results derived in Section \ref{sec:analysis}.
We fix the number of states $n = 10$ and time steps $k_f = 4$, and consider an LTV model. 
We create datasets and increase the number of experiments $q$ and study the method under increasing levels of noise.
The synthetic data sets are constructed using the formulation in \eqref{eq:LTVdyn}.
The entries of each $A[k]$ are drawn from a standard normal distribution and potentially scaled by a factor $\alpha_{A}$.
The input $u^{(j)}[k] \in \mathbb{R}^{n}$ is sparse, and acting on one randomly picked state variable at each time point in every experiment, hence $s = 1 \cdot k_f \cdot q = 4q$. 
The targeted node of each experiment is randomized, such that the nodes are perturbed uniformly throughout all experiments. 
The nonzero input entries are constructed randomly from a standard normal distribution and scaled by a factor $\alpha_{u}$.
$w^{(j)}[k] \in \mathbb{R}^{n}$ represents the noise acting on the state variables, sampled from a bounded distribution (either thresholded standard normal or uniform) and scaled by a factor $\alpha_{w}$.
In our experiments we have set $\alpha_{A} = 1$, $\alpha_u = 1$ and varied $\alpha_w$ for simulating different levels of noise that can be interpreted as noise percentage.
Following this setup, we have that the cardinality percentage of the input $\rho_{\pmb{u}} = \dfrac{s}{n k_f q} = \dfrac{1}{10}$. 
Lemma \ref{thm:sparsity} tells us that $\rho_{\pmb{u}} = \dfrac{1}{10} \le \dfrac{1}{2}$.

To assess the performance of our algorithm, we run a Monte Carlo experiment with $T$ iterations and compare the retrieved signal $\pmb{u}^*$ to the signal $\bar{\pmb{u}}$ that generated the data~$\pmb{z}$. We introduce two complementary metrics. 
First, to assess how many entries in $\pmb{u}^*$ were recovered correctly (regardless of magnitude), we compute the Mean Average Percentage Error of the cardinality ($MAPE_{card}$) by computing the number of false positives ($\# FP$) and false negatives ($\# FN$):
\begin{equation}
\begin{array}{rcl}
MAPE_{card} \left( \{ \pmb{u}^*(t) \}_{t=1}^T \right) &=& \dfrac{1}{T} \displaystyle \sum_{t=1}^T \frac{\# FP + \# FN}{n k_f q} \\
\end{array}
\end{equation}
Second, for the $s$ nonzero entries of the data-generating signal $\bar{\pmb{u}}$, coined $\bar{\pmb{u}}_{nz}$, we compute the Average Root Mean Square Error ($ARMSE_{nz}$) to assess the error in the magnitude:
\begin{equation}
\begin{array}{rcl}
ARMSE_{nz} \left( \{ \pmb{u}^*(t) \}_{t=1}^T \right) &=& \sqrt{\dfrac{1}{Tsq} \displaystyle \sum_{t=1}^T \| \bar{\pmb{u}}_{nz}(t) - \pmb{u}^*_{nz}(t)]  \|^2_2} \\
\end{array}
\end{equation}
For $\bar{\pmb{a}},\pmb{a}^*$, we simply track the standard $ARMSE$ over all entries, unless we try to retrieve a sparse solution. 

Figures \ref{fig:Results01} and \ref{fig:Results02} present the results of our synthetic experiments. In Figure \ref{fig:Results01} we took three different levels of measurement noise (0\%, 1\% and 5\%, equivalent to $\alpha_{w} \in \{ 0 , 0.01 , 0.05 \}$) and increased the number of experiments $q$.
In Figure \ref{fig:Results02} we fixed the number of experiments ($q = 30$), and increased the noise level.
\begin{figure}[h]
  \centering
  \epsfig{width=0.5\textwidth,file=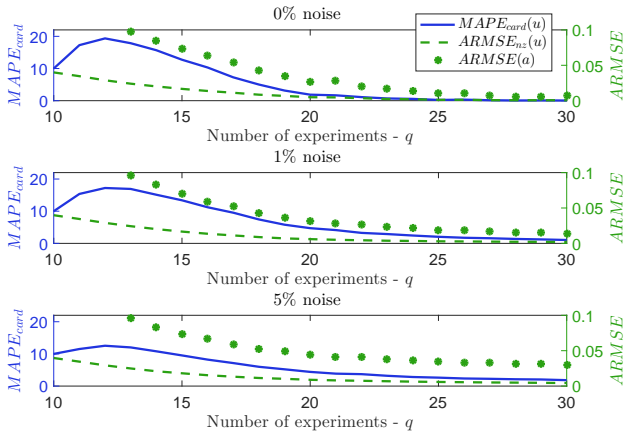}
  \caption{Average entry-wise error for three different levels of noise and increasing number of experiments.}
  \label{fig:Results01}
\end{figure}
\begin{figure}[h]
  \centering
  \epsfig{width=0.5\textwidth,file=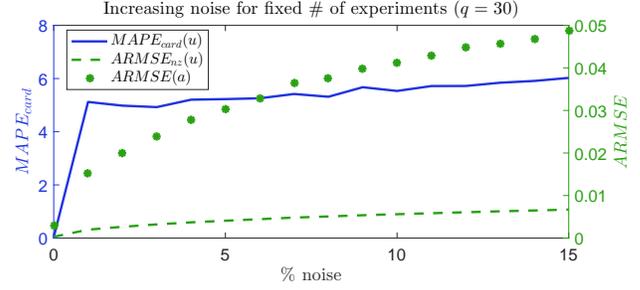}
  \caption{Average entry-wise error for increasing levels of noise and a fixed number of experiments ($q = 30$).}
  \label{fig:Results02}
\end{figure}
As anticipated by Section \ref{sec:analysis}, for the LTV case, we need $q > n$ to correctly retrieve the model parameters and input values. 
Across the different levels of noise, full convergence is reached between 20 to 30 experiments, depending on the necessary accuracy.
Since the unknowns are all drawn from a standard normal distribution, the ARMSE and MAPE metrics can be interpreted as a relative percentage-style error. We see that, for sufficient experiments $q$, the ARMSE and MAPE dive under the noise level added to the dynamics.
As expected, higher levels of noise yield a higher asymptotic error, which is clearly visualized in Figure \ref{fig:Results02}.
We see that the dynamical system model parameters in $\pmb{a}$ absorb an entry-wise error that is roughly a factor 5 higher than the entry-wise error in the nonzero input values of $\pmb{u}$.
Lastly, we notice that the ARMSE errors grow gradually and sublinearly, whereas the MAPE error jumps to 5.3\% for 1\% noise and then increases slightly to 6.2\% for 15\% noise.

\section{Conclusions}
\label{sec:conclusions}

This paper developed a method for blind system identification of discrete-time linear time-varying (LTV) models in settings where all states can be observed.
A sparse recovery problem was formulated to retrieve the dynamical system parameters and unknown input values by exploiting \emph{a priori} knowledge that the effect of unknown inputs is limited to affect a limited number of states and time points.
An optimization problem was formulated as a compressive sensing problem with a partially sparse support, which allowed analysis via sparse recovery theory.
This yielded sufficient conditions stating that the number of experiments should be greater than the number of states, less than half of the input values should be nonzero, and the matrices with system output measurements across time and experiments should be full rank.
Assessment of the method with synthetic data confirmed theoretical insights and provided further directions for designing experiments.
This work opens the further investigation of the blind system identification problem for time-varying systems. 
The authors are exploring the incorporation of more structured inputs to reduce the search space and make the overall recovery problem more efficient, and have applied the method to help design experiments in breast cancer research~\cite{dobbe_heterogeneity_2015}.

\begin{ack}                               
The authors would like to thank Young-Hwan Chang at OHSU and Margaret Chapman for their early involvement and useful feedback, and Joe Gray, James Korkola and Laura Heiser at OHSU for motivating this work through their breast cancer research and experiments. This work was supported by the \emph{Measuring, Modeling and Manipulating Heterogeneity} Grant from Oregon Health Sciences University (award number 1010517).  
\end{ack}

\printbibliography





\end{document}